\documentclass[sigconf,nonacm,natbib=false]{acmart}
\usepackage[sort,numbers]{natbib}
\usepackage{booktabs}
\usepackage{multirow}
\usepackage{algorithm}
\usepackage{algpseudocode}

\renewcommand\footnotetextcopyrightpermission[1]{}
\newtheorem{assumption}{Assumption}
\newtheorem{proposition}{Proposition}
\newtheorem{theorem}{Theorem}

\newcommand{\dtv}{d_{\mathrm{TV}}}

\begin{document}

\title{Taming Tail Risk in Financial Markets: Conformal Calibration for Nonstationary Portfolio VaR}

\author{Marc Schmitt}
\affiliation{\institution{University of Oxford}\city{}\country{}}
\email{marc.schmitt@cs.ox.ac.uk}

\begin{abstract}
Value-at-risk (VaR) forecasts drive trading constraints and capital allocation, yet realized exceedance rates concentrate in stress periods, when losses are largest. This paper studies sequential one-sided VaR calibration via conformal prediction. It proposes regime-weighted conformal calibration (RWC), which builds a safety buffer from past forecast errors using exponential time decay and regime-similarity weights. RWC is model-agnostic and wraps any conditional quantile forecaster to target a desired exceedance rate, with time-weighted calibration (TWC) as a special case. Coverage bounds are derived for arbitrary data-driven weights under smooth regime drift, without assuming weighted exchangeability. On the CRSP index and sixteen U.S. equity portfolios, RWC and TWC are benchmarked against modern online conformal methods at the Basel-relevant 99\% and 97.5\% levels. TWC is a strong default under drift, while regime weighting improves stress-period calibration for slowly adapting forecasters, and diagnostics indicate when localization is reliable.
\end{abstract}


\maketitle

\section{Introduction}
Forecasting and controlling tail risk in financial markets is notoriously difficult. Return distributions are heavy-tailed, exhibit volatility clustering, and shift over time; even a well-specified conditional model can degrade rapidly during crises or structural change. This nonstationarity is often \emph{regime-structured}: volatility and tail behavior cluster into persistent states (calm versus stress), as captured by regime-switching models \cite{hamilton1989,gray1996}; conditioning on the volatility state has first-order economic value \cite{moreira2017}. As a consequence, value-at-risk (VaR) forecasts are systematically miscalibrated, with realized exceedance rates deviating from nominal targets in stress periods \cite{berkowitz2002}. VaR-type constraints directly affect leverage and risk-taking \cite{basak2001,adrian2014}, and regulatory backtesting penalizes excess and clustered violations \cite{bcbs2019}. For example, under the Basel traffic-light framework, a bank that records too many VaR exceptions in its one-year backtest is moved into a higher penalty zone, and the multiplier applied to its market-risk capital charge rises accordingly.

Conformal prediction is attractive here because it wraps arbitrary black-box forecasters and provides finite-sample guarantees under exchangeability \cite{vovk2005}. Financial data, however, are sequential and nonstationary, so classical conformal validity degrades under time dependence and distribution shift. Two research threads address this. \emph{Weighted} conformal prediction reweights calibration scores, originally with density-ratio weights under covariate shift \cite{tibshirani2019}, and more generally with arbitrary weights at the price of a coverage gap governed by how far the data deviate from exchangeability \cite{barber2023}. \emph{Online} conformal prediction instead adapts a level or threshold by feedback on realized miscoverage: adaptive conformal inference (ACI) \cite{gibbs2021}, its parameter-free extension DtACI \cite{gibbs2024}, conformal PID control \cite{angelopoulos2023}, and related time-series methods \cite{xu2021,xu2023,zaffran2022,bastani2022}. Online methods guarantee long-run \emph{marginal} coverage under arbitrary shift; what they do not target is \emph{regime-conditional} stability: keeping exceedance rates near target within calm and stress states separately, rather than only on average. Since exact distribution-free conditional coverage is impossible \cite{barber2021}, practical regime-conditional calibration is necessarily approximate, and the relevant question is an empirical one: which calibration mechanisms deliver stable within-regime exceedance rates, at what cost in capital (bound tightness)?

We study this question for sequential one-sided VaR calibration. Our method, \emph{regime-weighted conformal calibration} (RWC), wraps any conditional quantile forecaster and calibrates an additive safety buffer as a weighted conformal quantile of past forecast errors, with weights combining exponential time decay (adaptation under drift) and a Gaussian kernel on simple regime features (realized volatility and mean absolute return), so that calibration borrows strength from comparable market conditions. The kernel-free special case, time-weighted conformal calibration (TWC), is a strong, computationally trivial default.

\textbf{Contributions.} (i)~This paper formalizes sequential one-sided VaR calibration under regime-structured nonstationarity and proposes regime-conditional calibration metrics aligned with practitioner backtests. (ii)~It gives an honest theoretical treatment for \emph{heuristic} (recency and similarity) weights: rather than assuming weighted exchangeability with respect to the algorithm's own weights, the analysis conditions on the regime path and bounds the coverage gap via total-variation smoothness of score distributions across regimes and time, building directly on nonexchangeable conformal theory \cite{barber2023}. The bound cleanly separates regime-mismatch bias, temporal-drift bias, and an effective-sample-size variance term, each observable through the reported diagnostics. (iii)~It benchmarks on the CRSP value-weighted index and sixteen CRSP-derived portfolios (1990--2024) against modern online conformal baselines (ACI, DtACI, conformal PID) and classical risk models (historical simulation, GARCH-$t$, gradient-boosting quantile regression as base forecasters), at both 99\% and 97.5\% VaR, with block-bootstrap inference and dynamic-quantile backtests. Regime weighting has a cost: when the base forecaster is already adaptive, time decay captures most of the drift. It delivers consistent stress-regime improvements for slowly adapting base models, however, and the diagnostics predict when localization helps.

\section{Related work}
\textbf{Conformal prediction and weighted extensions.} Split conformal prediction yields finite-sample marginal coverage under exchangeability \cite{vovk2005}; conformalized quantile regression adapts intervals to heteroscedasticity \cite{romano2019}. Weighted conformal prediction extends validity to covariate shift with density-ratio weights \cite{tibshirani2019}. \citet{barber2023} analyze conformal prediction \emph{beyond} exchangeability: with arbitrary fixed weights, coverage degrades by at most a weighted sum of total-variation distances between the realized data sequence and its swapped versions. Our guarantee instantiates this program for recency--similarity weights in a sequential one-sided setting, with the total-variation terms bounded by interpretable smoothness constants. Localized conformal prediction \cite{guan2023} and class-conditional variants \cite{ding2023} pursue related conditional targets in exchangeable settings; our regime weighting is a continuous, sequential analogue of such stratification.

\textbf{Online conformal methods.} ACI adjusts an effective miscoverage level by feedback \cite{gibbs2021}; DtACI aggregates ACI experts across learning rates and adapts to unknown shift sizes \cite{gibbs2024}; conformal PID control combines quantile tracking, error integration, and scorecasting \cite{angelopoulos2023}; further time-series approaches include EnbPI \cite{xu2021}, sequential predictive conformal inference \cite{xu2023}, aggregated ACI \cite{zaffran2022}, and multivalid prediction \cite{bastani2022}. These methods guarantee long-run marginal coverage under arbitrary distribution shift---a guarantee our weighted approach does not enjoy---but they adapt \emph{after} miscoverage materializes, whereas regime weighting conditions \emph{ex ante} on where the market currently sits. We therefore treat ACI, DtACI, and PID as primary baselines rather than alternatives to dismiss, and evaluate whether ex-ante regime localization adds value over ex-post feedback. Conformal risk control \cite{angelopoulos2024} generalizes conformal guarantees to monotone risk functionals; our target is one-sided quantile coverage, so we use ``calibration'' rather than ``risk control'' terminology throughout.

\textbf{Financial risk models and backtesting.} VaR is the canonical quantile risk measure \cite{rockafellar2000,engle2004}; systematic VaR failures in stress periods are well documented \cite{berkowitz2002}, and the classical remedies filter the volatility state, as in GARCH \cite{bollerslev1986} and conditional extreme-value theory \cite{mcneilfrey2000}. Recent ML-for-finance work targets tail risk via extreme-aware distributional reinforcement learning \cite{malekzadeh2024} and richer volatility forecasters \cite{kong2025}; our calibration layer is complementary, wrapping any such forecaster with a coverage guarantee. Standard backtests evaluate unconditional coverage and independence \cite{kupiec1995,christoffersen1998}, and the dynamic quantile test adds conditioning on past information \cite{engle2004}. We adopt these tests and add regime-stratified exceedance metrics, connecting the machine-learning calibration literature to risk-model validation practice.

\section{Problem setup}
We observe covariates and portfolio losses $(x_t, y_t) \in \mathcal{X} \times \mathbb{R}$, $t = 1, \dots, T$, where $x_t$ collects information available strictly before the loss $y_t$ is realized (for daily data, $y_t = -r_t$ with $r_t$ the close-to-close return, and all components of $x_t$ use data through day $t-1$). Fix a target exceedance level $\alpha \in (0,1)$ (e.g., $\alpha = 0.01$ for 99\% VaR). The object of estimation is a one-sided VaR bound: a data-driven function $U_t(\cdot)$ issued at the start of day $t$ such that
\begin{equation}
\mathbb{P}\big(y_t \le U_t(x_t)\big) \ge 1 - \alpha .
\label{eq:target}
\end{equation}
Under nonstationarity the distribution of $(x_t, y_t)$ varies with $t$, so we evaluate procedures by how closely realized exceedance rates track $\alpha$ over time \emph{and within market regimes}: for a regime feature $z_t = g(x_t) \in \mathbb{R}^d$ (below: realized volatility and mean absolute return) and bins $B_1,\dots,B_K$ of its range (below: volatility quintiles), we report binned exceedance rates $\widehat{\mathbb{P}}(y_t > U_t \mid z_t \in B_k)$, a sequential, kernel-smoothed relative of group-conditional coverage \cite{ding2023,guan2023}. The map $g$ is a fixed, observable feature transform---\emph{not} a latent regime label assumed known; we return to this distinction in Section~\ref{sec:limits}.

\textbf{Base forecaster and scores.} Let $\hat{q}_t = f_t(x_t)$ be any forecast of the conditional $(1-\alpha)$-quantile of $y_t$ from a model trained on past data (we use historical simulation, GARCH-$t$, and gradient-boosting quantile regression). Conformal calibration corrects the base forecast additively, $U_t = \hat{q}_t + \hat{c}_t$, where the buffer $\hat{c}_t$ is calibrated from past one-sided conformity scores
\begin{equation}
s_i := y_i - \hat{q}_i , \qquad i < t,
\label{eq:score}
\end{equation}
so that large positive $s_i$ records that the base model underpredicted tail risk at time $i$.

\section{Method: regime-weighted conformal calibration}
\label{sec:method}
At time $t$, let $\mathcal{I}_t = \{\max(1, t-m), \dots, t-1\}$ be the calibration window ($m$ = buffer size). RWC assigns each $i \in \mathcal{I}_t$ the weight
\begin{equation}
w_i(t) \;=\; \underbrace{e^{-\lambda (t-i)}}_{\text{recency}} \;\cdot\; \underbrace{\exp\!\Big( -\tfrac{\lVert z_i - z_t \rVert^2}{2h^2} \Big)}_{\text{regime similarity}} ,
\label{eq:weights}
\end{equation}
with decay rate $\lambda \ge 0$ and kernel bandwidth $h > 0$; the current point receives $w_t(t) = e^{0} K_h(z_t, z_t) = 1$, consistent with \eqref{eq:weights} evaluated at $i = t$. Coordinates of $z_t$ are standardized using pre-validation statistics, so $h$ is in standard-deviation units. Setting $h = \infty$ (kernel $\equiv 1$) recovers TWC; additionally setting $\lambda = 0$ recovers a sliding-window conformal (SWC) baseline.

\textbf{Corrected weighted quantile.} Let $W_t = \sum_{i \in \mathcal{I}_t} w_i(t)$ and $\tilde{w}_i(t) = w_i(t)/W_t$. RWC sets
\begin{equation}
\hat{c}_t = Q^{\tilde{w}(t)}_{\rho_t}\big(\{ s_i \}_{i \in \mathcal{I}_t}\big),
\qquad
\rho_t = \min\!\Big\{ 1, \, (1-\alpha)\Big(1 + \tfrac{1}{W_t}\Big) \Big\},
\label{eq:quantile}
\end{equation}
the weighted $\rho_t$-quantile of past scores (smallest $c$ whose cumulative weight reaches $\rho_t$). The inflation of the level from $1-\alpha$ to $\rho_t$ is the weighted analogue of the $(n+1)$ correction in split conformal prediction \cite{tibshirani2019} and is used \emph{throughout}---in the theory and in every experiment. Finally $U_t = \hat{q}_t + \hat{c}_t$.

\textbf{Effective sample size safeguard.} Kernel localization concentrates weight and can starve the calibration set. With $n_{\mathrm{eff}}(t) = 1/\sum_i \tilde{w}_i(t)^2$, if $n_{\mathrm{eff}}(t) < n_{\min}$ we drop the kernel for that step (fall back to TWC weights). We report $n_{\mathrm{eff}}$ and the effective memory $\tau_t = \sum_i \tilde{w}_i(t)\,(t-i)$ as diagnostics: they are the empirical quantities that appear in the variance and drift-bias terms of Theorem~\ref{thm:gap}.

\begin{algorithm}[t]
\caption{RWC: regime-weighted conformal VaR calibration}
\label{alg:rwc}
\begin{algorithmic}[1]
\Require target $\alpha$; base forecaster $f_t$; window $m$; decay $\lambda$; bandwidth $h$; regime map $g$; ESS floor $n_{\min}$
\For{$t = t_0, t_0+1, \dots$}
  \State observe $x_t$; set $z_t = g(x_t)$; predict $\hat{q}_t = f_t(x_t)$
  \State compute weights $w_i(t)$ by \eqref{eq:weights} for $i \in \mathcal{I}_t$
  \If{$n_{\mathrm{eff}}(t) < n_{\min}$} drop kernel (TWC weights) \EndIf
  \State $\hat{c}_t \gets$ corrected weighted quantile \eqref{eq:quantile}; issue $U_t = \hat{q}_t + \hat{c}_t$
  \State observe $y_t$; append score $s_t = y_t - \hat{q}_t$ and $z_t$ to buffers
\EndFor
\end{algorithmic}
\end{algorithm}

\section{Theory: coverage with heuristic weights}
\label{sec:theory}
Classical weighted conformal analysis assumes \emph{weighted exchangeability}: the joint density of the calibration-plus-test collection factorizes as a weight function times a symmetric function, with the algorithm's weights equal to the true density ratios \cite{tibshirani2019}. For recency--similarity weights such as \eqref{eq:weights} this premise is circular---it amounts to assuming the data obey exactly the localization the algorithm imposes---so we do not adopt it as our working assumption. We state the classical result once as a benchmark, then give our main guarantee, which holds for \emph{arbitrary} weights.

\begin{proposition}[Benchmark; \citealp{tibshirani2019}]
\label{prop:wex}
Fix $t$ and condition on $(x_t, z_t)$. If $\{(s_i, w_i(t))\}_{i \in \mathcal{I}_t} \cup \{(s_t, 1)\}$ is weighted exchangeable, then $U_t$ from \eqref{eq:quantile} satisfies $\mathbb{P}(y_t \le U_t \mid x_t) \ge 1-\alpha$.
\end{proposition}

Weighted exchangeability with respect to \eqref{eq:weights} will not hold exactly for financial data; the value of Proposition~\ref{prop:wex} is to fix the corrected level $\rho_t$. Our main result quantifies what is lost when it fails.

\begin{assumption}[Smooth regime drift in total variation]
\label{ass:tv}
Let $P_i(\cdot \mid z)$ denote the conditional law of the score $s_i$ given $z_i = z$ and the past. There exist $L_z, L_t \ge 0$ such that for all $i \le t$ and all $z, z'$ in the (bounded) support of the regime features,
\[
\dtv\big( P_i(\cdot \mid z), \, P_t(\cdot \mid z') \big) \;\le\; L_z \lVert z - z' \rVert \; + \; L_t\,(t - i).
\]
\end{assumption}

\begin{assumption}[Conditional independence given the regime path]
\label{ass:ci}
Conditional on the regime path $Z_t = (z_i)_{i \in \mathcal{I}_t \cup \{t\}}$, the scores $\{s_i\}_{i \in \mathcal{I}_t \cup \{t\}}$ are independent, with $s_i \sim P_i(\cdot \mid z_i)$.
\end{assumption}

Assumption~\ref{ass:tv} formalizes ``similar regimes and nearby dates have similar error distributions'' directly in the metric the theory consumes; it is a smoothness condition on \emph{conditional} laws and does not restrict the (arbitrarily nonstationary) path of $z_t$ itself. Assumption~\ref{ass:ci} makes the statement clean; serial dependence of scores beyond what $z$ captures weakens the result only through the concentration constants (remark below).

\begin{theorem}[Coverage gap for regime-weighted calibration]
\label{thm:gap}
Fix $t$, condition on the regime path $Z_t$, and let the weights \eqref{eq:weights} (which are $Z_t$-measurable) be arbitrary otherwise. Under Assumptions~\ref{ass:tv}--\ref{ass:ci}, the RWC bound $U_t$ from \eqref{eq:quantile} satisfies
\[
\mathbb{P}\big( y_t \le U_t \,\big|\, Z_t \big) \;\ge\; 1 - \alpha \;-\; \varepsilon_t,
\]
\[
\varepsilon_t \;=\; 2 \sum_{i \in \mathcal{I}_t} \bar{w}_i(t)\, \dtv\big( P_i(\cdot \mid z_i), P_t(\cdot \mid z_t) \big),
\]
where $\bar{w}_i(t) = w_i(t) / (W_t + 1)$. Consequently, by Assumption~\ref{ass:tv},
\begin{align*}
\varepsilon_t \;&\le\; 2 L_z \sum_{i} \bar{w}_i(t) \lVert z_i - z_t \rVert \; + \; 2 L_t \sum_{i} \bar{w}_i(t)\,(t-i) \\
&=\; O(L_z h) + O(L_t \tau_t) ,
\end{align*}
the second equality holding when the kernel concentrates its normalized mass on $\{\lVert z_i - z_t\rVert = O(h)\}$, with $\tau_t$ the effective memory of Section~\ref{sec:method}.
\end{theorem}

\begin{proof}
Conditional on $Z_t$ the weights are fixed constants, and by Assumption~\ref{ass:ci} the scores are independent draws from $P_i(\cdot\mid z_i)$. Theorem~2a of \citet{barber2023} applied to the one-sided score sequence $S = (s_i)_{i\in\mathcal{I}_t\cup\{t\}}$ with fixed weights $\{w_i(t)\}\cup\{1\}$ gives
$\mathbb{P}(y_t \le U_t \mid Z_t) \ge 1-\alpha - \sum_i \bar{w}_i(t)\, \dtv(S, S^{(i)})$,
where $S^{(i)}$ is the sequence with entries $i$ and $t$ swapped. Under independence, $S$ and $S^{(i)}$ are product laws that differ only in coordinates $i$ and $t$, where they carry $P_i \otimes P_t$ and $P_t \otimes P_i$ respectively (writing $P_i = P_i(\cdot\mid z_i)$, $P_t = P_t(\cdot\mid z_t)$; the common factors cancel). Two triangle inequalities through the intermediate law $P_t \otimes P_t$ give
\begin{align*}
\dtv\big(P_i \otimes P_t, \, P_t \otimes P_i\big)
&\le \dtv\big(P_i \otimes P_t, P_t \otimes P_t\big) \\
&\quad + \dtv\big(P_t \otimes P_t, P_t \otimes P_i\big)
= 2\,\dtv(P_i, P_t),
\end{align*}
using that $\dtv(P\otimes R, Q\otimes R) = \dtv(P, Q)$. (The factor $2$ is not removable in general: for $P_i$ uniform on $\{1,2\}$ and $P_t$ uniform on $\{2,3\}$, $\dtv(P_i, P_t) = \tfrac12$ while the swapped products are at distance $\tfrac34$.) Applying Assumption~\ref{ass:tv} to each $\dtv(P_i, P_t)$ term and, for the final display, the kernel concentration of normalized weights and the definition of $\tau_t$, yields the stated bounds.
\end{proof}

\textbf{Reading the bound.} The gap $\varepsilon_t$ decomposes into a \emph{regime-mismatch} term $O(L_z h)$, shrunk by localizing the kernel, and a \emph{drift} term $L_t \tau_t$, shrunk by recency decay; both are computable up to constants from the diagnostics $(n_{\mathrm{eff}}, \tau_t)$ we report. Localization is not free: shrinking $h$ or raising $\lambda$ reduces $n_{\mathrm{eff}}(t)$ and inflates the finite-sample noise of the weighted quantile, which is $O(n_{\mathrm{eff}}(t)^{-1/2})$ for weighted averages of bounded indicators under Assumption~\ref{ass:ci} (and degrades gracefully under mixing in place of independence). This bias--variance tradeoff is exactly what the ESS safeguard controls and what the bandwidth sweep in Section~\ref{sec:experiments} traces empirically.

\textbf{What the theory does and does not claim.} Theorem~\ref{thm:gap} conditions on the regime path---a regime-conditional statement aligned with the binned exceedance metrics we report---but the gap $\varepsilon_t$ does not vanish for fixed $(\lambda, h)$: RWC does not inherit the distribution-free long-run marginal coverage that feedback methods (ACI, DtACI, PID) achieve by construction \cite{gibbs2021,gibbs2024,angelopoulos2023}, a consequence of the impossibility of exact distribution-free conditional coverage \cite{barber2021}. The empirical question is therefore whether ex-ante localization buys regime-conditional stability that ex-post feedback does not, and at what capital cost; Section~\ref{sec:experiments} answers with direct comparisons.

\section{Experiments}
\label{sec:experiments}
\textbf{Data.} We use daily U.S. equity data, 1990-03-30 to 2024-12-31 (8{,}755 trading days). The headline series is the CRSP value-weighted market index (via WRDS), the benchmark survivorship-bias-free record in empirical asset pricing. To address breadth beyond a single index, we add sixteen CRSP-derived daily portfolios from Ken French's data library: ten value-weighted industry portfolios and six size/book-to-market (2$\times$3) portfolios. Losses are $y_t = -r_t$. Splits are chronological: training to 2011-01-31, validation 2011-02-01 to 2018-01-16 (hyperparameter tuning), test 2018-01-17 to 2024-12-31 ($N = 1{,}751$ days spanning the 2018 volatility spike, the 2020 COVID crash, the 2022 tightening cycle, and the 2023--24 recovery).

\textbf{Base forecasters.} (i) \emph{HS}: rolling 500-day historical simulation; (ii) \emph{GARCH}: GARCH(1,1) \cite{bollerslev1986} with Student-$t$ innovations, refit every 21 days on a rolling 10-year window; (iii) \emph{GBDT}: gradient-boosting quantile regression (LightGBM \cite{ke2017}) on lagged return, volatility, and range features, refit every 21 days. The three span the adaptivity spectrum from static to flexible, which turns out to be the key moderator of regime weighting's value.

\textbf{Calibrators.} SWC (sliding window), TWC, RWC (Algorithm~\ref{alg:rwc}), ACI \cite{gibbs2021}, DtACI \cite{gibbs2024} with their published candidate set of eight learning rates ($\gamma_i = 0.001 \cdot 2^{i-1}$) and default aggregation parameters, and conformal PID \cite{angelopoulos2023} (quantile tracking $\eta_t = 0.1 \hat{B}_t$ plus tangent integrator with the recommended saturation constants; no scorecaster, to keep base-model attribution clean). Regime features are $z_t = (\mathrm{RV}^{21}_t, \mathrm{MAR}^5_t)$: 21-day realized volatility and 5-day mean absolute return, both computed through day $t-1$ and standardized on pre-validation statistics. $(m, \lambda, h)$ and ACI's $\gamma$ are tuned on the validation period by minimizing $|\widehat{\mathrm{Exc}} - \alpha| + 0.5 \max(0, \widehat{\mathrm{RollMax}} - \alpha)$; DtACI and PID require no tuning. Hyperparameters tuned on the index are applied unchanged to all sixteen portfolios, mimicking deployment.

\textbf{Metrics.} Test-period exceedance rate with 95\% moving-block bootstrap confidence intervals (63-day blocks, 2{,}000 draws); average bound level (average VaR, in basis points---the capital cost of the bound); exceedance by realized-volatility quintile with the stability summaries $\mathrm{Reg\text{-}MAE} = \frac{1}{5}\sum_k |\hat{e}_k - \alpha|$ and $\mathrm{Reg\text{-}MaxDev}$; Kupiec, Christoffersen, and Engle--Manganelli dynamic-quantile (DQ) backtests \cite{kupiec1995,christoffersen1998,engle2004}.

% ---- results tables/figures inserted here ----
\begin{table}[t]
\caption{Test-period 99\% and 97.5\% VaR calibration on the CRSP value-weighted index (2018--2024, $N{=}1751$). Exc.\ = exceedance rate (\%; target $100\alpha$); VaR = average bound (bps; capital cost). $\dagger$: Kupiec unconditional-coverage rejection at 5\%; $\ddagger$: Engle--Manganelli DQ rejection at 5\%.}
\label{tab:index}
\centering\footnotesize
\setlength{\tabcolsep}{2.6pt}
\begin{tabular}{l cc cc cc}
\toprule
 & \multicolumn{2}{c}{HS} & \multicolumn{2}{c}{GARCH-$t$} & \multicolumn{2}{c}{GBDT} \\
\cmidrule(lr){2-3}\cmidrule(lr){4-5}\cmidrule(lr){6-7}
Method & Exc.\ (\%) & VaR & Exc.\ (\%) & VaR & Exc.\ (\%) & VaR \\
\midrule
\multicolumn{7}{l}{\textit{A. 99\% VaR ($\alpha=1\%$)}} \\
Base & $1.83^{\dagger\ddagger}$ & 341 & $1.94^{\dagger\ddagger}$ & 260 & $4.05^{\dagger\ddagger}$ & 198 \\
SWC & $1.71^{\dagger\ddagger}$ & 430 & $0.97$ & 329 & $0.91^{\ddagger}$ & 372 \\
ACI & $1.48^{\ddagger}$ & 457 & $1.14$ & 330 & $1.14$ & 379 \\
DtACI & $1.88^{\dagger\ddagger}$ & 325 & $1.54^{\dagger\ddagger}$ & 299 & $1.66^{\dagger\ddagger}$ & 299 \\
PID & $1.48^{\ddagger}$ & 365 & $1.09^{\ddagger}$ & 342 & $1.26^{\ddagger}$ & 308 \\
TWC & $0.86^{\ddagger}$ & 493 & $0.86$ & 325 & $0.91^{\ddagger}$ & 375 \\
RWC & $0.63^{\ddagger}$ & 521 & $0.63$ & 381 & $0.80^{\ddagger}$ & 313 \\
\midrule
\multicolumn{7}{l}{\textit{B. 97.5\% VaR ($\alpha=2.5\%$)}} \\
Base & $3.71^{\dagger\ddagger}$ & 242 & $3.88^{\dagger\ddagger}$ & 200 & $6.80^{\dagger\ddagger}$ & 165 \\
SWC & $2.68^{\ddagger}$ & 272 & $2.57^{\ddagger}$ & 242 & $2.40^{\ddagger}$ & 241 \\
ACI & $3.14^{\ddagger}$ & 269 & $2.68^{\ddagger}$ & 245 & $2.63^{\ddagger}$ & 269 \\
DtACI & $3.54^{\dagger\ddagger}$ & 246 & $2.97^{\ddagger}$ & 233 & $2.91^{\ddagger}$ & 230 \\
PID & $3.08^{\ddagger}$ & 248 & $2.86^{\ddagger}$ & 253 & $2.74^{\ddagger}$ & 255 \\
TWC & $2.46^{\ddagger}$ & 307 & $2.40^{\ddagger}$ & 251 & $2.06^{\ddagger}$ & 261 \\
RWC & $2.06^{\ddagger}$ & 340 & $2.34$ & 258 & $2.51^{\ddagger}$ & 231 \\
\bottomrule
\end{tabular}
\end{table}
\begin{table}[t]
\caption{Exceedance rate (\%) by realized-volatility quintile (Q1 calm $\to$ Q5 stress), 99\% VaR, CRSP index test period. Reg-MAE = mean absolute deviation from the 1\% target across quintiles (pp).}
\label{tab:regime}
\centering\footnotesize
\setlength{\tabcolsep}{3.6pt}
\begin{tabular}{l rrrrr r}
\toprule
Method & Q1 & Q2 & Q3 & Q4 & Q5 & Reg-MAE \\
\midrule
\multicolumn{7}{l}{\textit{HS base}} \\
SWC & 1.71 & 0.57 & 0.86 & 2.29 & 3.14 & 0.94 \\
ACI & 1.71 & 0.57 & 0.57 & 2.29 & 2.29 & 0.83 \\
DtACI & 3.13 & 1.14 & 0.57 & 2.57 & 2.00 & 1.06 \\
PID & 2.28 & 0.57 & 0.86 & 2.00 & 1.71 & 0.71 \\
TWC & 0.85 & 0.57 & 0.29 & 1.43 & 1.14 & 0.37 \\
RWC & 1.42 & 0.57 & 0.00 & 0.29 & 0.86 & 0.54 \\
\midrule
\multicolumn{7}{l}{\textit{GARCH-$t$ base}} \\
SWC & 2.85 & 0.86 & 0.29 & 0.29 & 0.57 & 0.77 \\
ACI & 2.85 & 1.43 & 0.29 & 0.86 & 0.29 & 0.77 \\
DtACI & 3.99 & 1.14 & 0.57 & 1.43 & 0.57 & 0.88 \\
PID & 2.56 & 1.14 & 0.29 & 0.86 & 0.57 & 0.60 \\
TWC & 2.28 & 0.86 & 0.29 & 0.29 & 0.57 & 0.66 \\
RWC & 1.42 & 0.57 & 0.29 & 0.29 & 0.57 & 0.54 \\
\midrule
\multicolumn{7}{l}{\textit{GBDT base}} \\
SWC & 1.42 & 0.86 & 0.57 & 0.57 & 1.14 & 0.31 \\
ACI & 1.71 & 1.14 & 0.86 & 0.86 & 1.14 & 0.26 \\
DtACI & 3.42 & 1.14 & 0.86 & 1.43 & 1.43 & 0.71 \\
PID & 2.85 & 0.86 & 0.57 & 0.86 & 1.14 & 0.54 \\
TWC & 1.42 & 0.86 & 0.57 & 0.57 & 1.14 & 0.31 \\
RWC & 1.14 & 0.57 & 0.57 & 0.57 & 1.14 & 0.31 \\
\bottomrule
\end{tabular}
\end{table}
\begin{figure*}[t]
\centering
\includegraphics[width=0.92\textwidth]{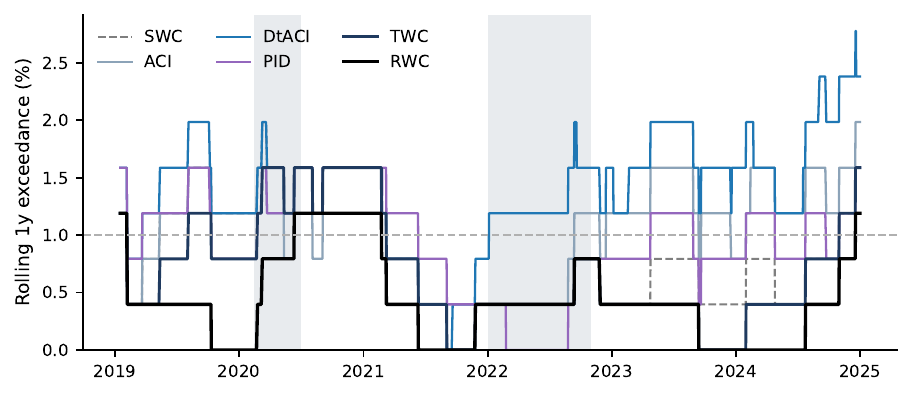}
\caption{Rolling one-year exceedance rate, 99\% VaR on the CRSP value-weighted index (GARCH base). Dashed line: 1\% target. Shading: COVID-19 crash and 2022 tightening cycle.}
\label{fig:rolling}
\end{figure*}
\begin{table*}[t]
\caption{Cross-sectional results: sixteen CRSP-derived portfolios (ten industries, six size/book-to-market), test period, hyperparameters tuned once on the index. $|$Exc$-\alpha|$ = mean absolute calibration error across portfolios (pp); Q5 = mean stress-quintile exceedance (\%); VaR = mean average bound (bps); DQ = number of portfolios (of 16) passing the dynamic-quantile backtest at 5\%.}
\label{tab:panel}
\centering\small
\setlength{\tabcolsep}{6.5pt}
\begin{tabular}{l rrrr rrrr rrrr}
\toprule
 & \multicolumn{4}{c}{HS} & \multicolumn{4}{c}{GARCH-$t$} & \multicolumn{4}{c}{GBDT} \\
\cmidrule(lr){2-5}\cmidrule(lr){6-9}\cmidrule(lr){10-13}
Method & $|$E$-\alpha|$ & Q5 & VaR & DQ & $|$E$-\alpha|$ & Q5 & VaR & DQ & $|$E$-\alpha|$ & Q5 & VaR & DQ \\
\midrule
\multicolumn{13}{l}{\textit{A. 99\% VaR}} \\
Base & 0.68 & 3.55 & 389 & 0 & 0.72 & 0.89 & 323 & 4 & 3.32 & 4.93 & 244 & 0 \\
SWC & 0.53 & 2.36 & 492 & 0 & 0.21 & 0.57 & 403 & 14 & 0.11 & 1.20 & 445 & 2 \\
ACI & 0.33 & 1.64 & 521 & 0 & 0.11 & 0.54 & 399 & 13 & 0.10 & 1.12 & 450 & 5 \\
DtACI & 0.83 & 1.71 & 381 & 0 & 0.56 & 0.64 & 357 & 2 & 0.63 & 1.30 & 366 & 2 \\
PID & 0.34 & 1.71 & 434 & 0 & 0.13 & 0.59 & 390 & 10 & 0.17 & 1.30 & 387 & 6 \\
TWC & 0.22 & 1.21 & 554 & 0 & 0.16 & 0.59 & 388 & 13 & 0.14 & 1.04 & 447 & 6 \\
RWC & 0.40 & 0.82 & 627 & 1 & 0.47 & 0.45 & 463 & 13 & 0.10 & 1.30 & 392 & 4 \\
\midrule
\multicolumn{13}{l}{\textit{B. 97.5\% VaR}} \\
Base & 0.92 & 6.70 & 285 & 0 & 1.24 & 2.34 & 254 & 0 & 4.16 & 7.38 & 207 & 0 \\
SWC & 0.14 & 4.61 & 322 & 0 & 0.12 & 1.61 & 296 & 10 & 0.19 & 2.77 & 299 & 2 \\
ACI & 0.41 & 2.05 & 323 & 0 & 0.11 & 1.48 & 298 & 9 & 0.09 & 1.86 & 329 & 2 \\
DtACI & 0.71 & 3.12 & 290 & 0 & 0.45 & 1.46 & 286 & 6 & 0.48 & 2.20 & 291 & 5 \\
PID & 0.41 & 2.55 & 303 & 0 & 0.24 & 1.27 & 308 & 2 & 0.26 & 1.82 & 311 & 1 \\
TWC & 0.45 & 2.29 & 363 & 0 & 0.16 & 1.46 & 303 & 12 & 0.30 & 2.09 & 323 & 6 \\
RWC & 0.78 & 1.91 & 417 & 0 & 0.32 & 1.27 & 315 & 13 & 0.20 & 2.86 & 291 & 2 \\
\bottomrule
\end{tabular}
\end{table*}
\begin{figure*}[t]
\centering
\includegraphics[width=0.92\textwidth]{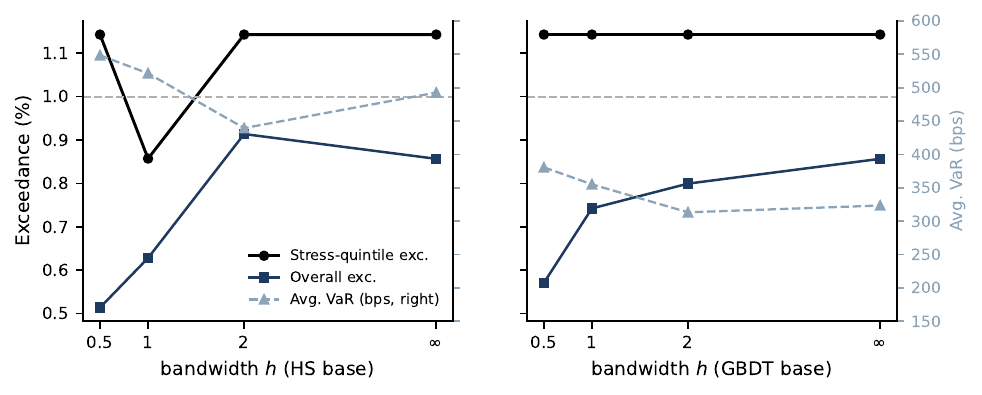}
\caption{Bandwidth sweep at tuned $(m, \lambda)$, 99\% VaR: overall and stress-quintile exceedance (left axis) and average bound (right axis) as regime localization strengthens ($h \downarrow$); $h = \infty$ is the time-weighted limit.}
\label{fig:sweep}
\end{figure*}

\subsection{Index results}
Table~\ref{tab:index} reports test-period calibration and capital cost. Uncalibrated base forecasters miss badly at the 99\% level---GBDT exceeds at $4.05\%$, GARCH-$t$ at $1.94\%$, HS at $1.83\%$ against a $1\%$ target, all rejected by Kupiec and DQ tests---and every conformal wrapper restores exceedance to the neighborhood of the target. Feedback methods track the marginal target closely (ACI: $1.14\%$ on GARCH and GBDT); weighted-quantile methods run deliberately conservative ($\mathrm{TWC}$: $0.86$--$0.91\%$; $\mathrm{RWC}$: $0.63$--$0.80\%$), the visible price of the finite-sample correction $\rho_t$ in \eqref{eq:quantile}. DtACI overshoots at the extreme tail ($1.54$--$1.89\%$, with Kupiec rejections on all three bases): with 252-day calibration windows, the $1\%$ quantile is granular and the expert aggregation reacts to noise; at $\alpha = 2.5\%$ (panel~B) it is well calibrated. Conformal PID lands between. Block-bootstrap $95\%$ intervals show the limits of single-index inference ($[0.17, 0.80]\%$ for RWC vs.\ $[0.69, 1.37]\%$ for ACI, GARCH-$t$ base), motivating the sixteen-portfolio replication of Section~\ref{sec:panel}. At $\alpha = 2.5\%$ (panel~B), the GARCH-based RWC is the only configuration that passes the DQ test ($p = 0.15$). Capital cost separates the weighted methods: on the GBDT base, RWC delivers its coverage at $313$ bps average VaR versus $375$ bps for separately tuned TWC ($17\%$ tighter) and $323$ bps for the time-weighted limit at matched $(m, \lambda)$---the kernel keeps distant-but-similar days relevant, letting the tuner exploit a longer window ($m{=}756$, median $n_{\mathrm{eff}} = 594$ vs.\ $247$ for TWC). On the HS base the sign flips (RWC $521$ vs.\ TWC $493$ bps): localization spends capital to fix stress-regime coverage that the static base model misses. Backtests add a structural lesson: with the GARCH-$t$ base, every calibrator except DtACI passes unconditional coverage and all but DtACI and PID pass the DQ test, while \emph{no} calibrator rescues the HS base from DQ rejection---violations cluster because the base cannot adapt within regimes, and an additive buffer cannot undo that. A conformal layer complements a conditional risk model; it does not substitute for one.

\subsection{Regime-conditional calibration}
Table~\ref{tab:regime} stratifies exceedances by realized-volatility quintile; this is where ex-ante localization and ex-post feedback separate. On the HS base, feedback methods restore \emph{average} coverage by construction but concentrate their errors in stress: top-quintile exceedance is $2.29\%$ (ACI), $2.00\%$ (DtACI), and $1.71\%$ (PID), against $1.14\%$ for TWC and $0.86\%$ for RWC. The mechanism is visible in Figure~\ref{fig:rolling}: feedback methods widen only \emph{after} a run of violations, so each regime transition is paid for in stress-period exceedances, while recency- and regime-weighted calibration repositions the buffer as soon as the regime features move. On the GARCH-$t$ base the base model itself tracks volatility, all methods keep the stress quintile at or below target ($0.29$--$0.57\%$), and regime weighting is unnecessary---consistent with Theorem~\ref{thm:gap}, since the score distribution is then nearly regime-homogeneous ($L_z$ small) and localization only costs effective sample size. On the GBDT base a residual stress undercoverage ($1.14\%$) is shared by all weighted methods; RWC's contribution there is not coverage but capital, as above. Across bases, TWC or RWC attain the best regime-stability summaries on HS (Reg-MAE $0.37$/$0.54$ pp), while ACI is best on GBDT ($0.26$ pp): when the base forecaster already absorbs regime structure, lightweight feedback is enough; when it does not, ex-ante localization is what keeps stress-period risk honest.

\subsection{Cross-sectional evidence: sixteen portfolios}
\label{sec:panel}
Table~\ref{tab:panel} asks whether the index findings survive contact with a cross-section: sixteen portfolios, hyperparameters frozen at the index-tuned values (no per-portfolio tuning). They do, on all three fronts. \emph{Stress calibration:} on the HS base, mean stress-quintile exceedance is $0.82\%$ for RWC versus $1.21\%$ (TWC), $1.64\%$ (ACI), $1.71\%$ (DtACI and PID), and $2.36\%$ (SWC)---ex-ante localization halves stress-period miscoverage relative to feedback methods, at a capital premium ($627$ vs.\ $554$ bps for TWC). \emph{Capital efficiency:} on the GBDT base, RWC attains the (tied-)best mean calibration error ($0.10$ pp) at the lowest average bound among well-calibrated methods ($392$ bps vs.\ $447$ for TWC and $450$ for ACI), a $12\%$ capital saving that compounds across a book. \emph{Backtest quality:} with the GARCH-$t$ base, SWC/TWC/RWC/ACI pass the DQ test on $13$--$14$ of $16$ portfolios; with the HS base, at most one portfolio passes under \emph{any} calibrator, and DtACI's index-level overshoot persists ($0.56$--$0.83$ pp mean error, $\le 2$ DQ passes). The consistency of all three patterns under frozen hyperparameters is the deployment-relevant result: the method transfers across assets without retuning.

\subsection{Localization diagnostics and sensitivity}
Figure~\ref{fig:sweep} traces the localization--variance tradeoff of Theorem~\ref{thm:gap} empirically. On the HS base, tightening the kernel from $h = \infty$ to $h = 1$ cuts stress-quintile exceedance to $0.86\%$ while average VaR rises toward $521$ bps ($h = 0.5$: $548$ bps)---capital buys regime-mismatch bias reduction, the $O(L_z h)$ term at work. On the GBDT base, stress exceedance is flat in $h$ while capital \emph{falls} from $381$ bps ($h{=}0.5$) to $313$ bps ($h{=}2$): localization prunes stale calm-period scores without changing tail coverage. The recency rate is the more powerful knob: raising $\lambda$ from $0.002$ to $0.01$ at tuned $m$ drives the HS-base bound to $0.29\%$ exceedance at $744$ bps---rapid decay shrinks $n_{\mathrm{eff}}$, and the corrected level $\rho_t$ responds by inflating the buffer. Tuned configurations sit in a comfortable region: median $n_{\mathrm{eff}}$ between $138$ and $638$, effective memory $\tau_t$ between $73$ and $287$ days, and tenth-percentile $n_{\mathrm{eff}} \ge 87$, so the ESS safeguard binds rarely and practitioners can monitor both quantities in production.

\section{Limitations}
\label{sec:limits}
Our guarantees hold conditional on smoothness (Assumption~\ref{ass:tv}) and conditional independence (Assumption~\ref{ass:ci}) that are approximations for financial data; the coverage gap does not vanish for fixed $(\lambda, h)$, and long-run marginal coverage is better served by feedback methods, three of which we benchmark. The regime map $g$ is a fixed feature transform; learned or latent regime labels (e.g., HMM-filtered states) are a natural extension, and mixed empirical results across base forecasters suggest the value of localization depends on what the base model already captures. We study one-step-ahead risk for liquid equity portfolios; multi-day horizons, illiquid assets, and feedback from risk constraints to prices introduce dependence we do not model. CRSP is proprietary but ubiquitous in academic finance, and the portfolio data are publicly available.

\section{Conclusion}
This paper revisits sequential one-sided VaR calibration under regime-structured nonstationarity. A single weighted-conformal mechanism (exponential recency decay plus an optional regime-similarity kernel with an ESS safeguard) provides a model-agnostic calibration layer whose coverage gap is bounded for arbitrary heuristic weights, without assuming weighted exchangeability. Empirically, across three base forecasters, the index and sixteen portfolios, and two VaR levels, time-decay calibration is a robust default; regime weighting adds value where base models adapt slowly, concentrated in stress regimes, and simple diagnostics ($n_{\mathrm{eff}}$, $\tau_t$) flag when localization is trustworthy. Conformal calibration layers are a practical reliability component between forecasting models and the risk constraints they feed.

\bibliographystyle{ACM-Reference-Format}
\bibliography{references}

@book{vovk2005,
  author = {Vovk, Vladimir and Gammerman, Alexander and Shafer, Glenn},
  title = {Algorithmic Learning in a Random World},
  publisher = {Springer},
  address = {New York},
  year = {2005}
}

@inproceedings{romano2019,
  author = {Romano, Yaniv and Patterson, Evan and Cand\`es, Emmanuel J.},
  title = {Conformalized Quantile Regression},
  booktitle = {Advances in Neural Information Processing Systems},
  volume = {32},
  pages = {3538--3548},
  year = {2019}
}

@inproceedings{tibshirani2019,
  author = {Tibshirani, Ryan J. and Barber, Rina Foygel and Cand\`es, Emmanuel J. and Ramdas, Aaditya},
  title = {Conformal Prediction Under Covariate Shift},
  booktitle = {Advances in Neural Information Processing Systems},
  volume = {32},
  year = {2019}
}

@article{barber2021,
  author = {Barber, Rina Foygel and Cand\`es, Emmanuel J. and Ramdas, Aaditya and Tibshirani, Ryan J.},
  title = {The Limits of Distribution-Free Conditional Predictive Inference},
  journal = {Information and Inference: A Journal of the IMA},
  volume = {10},
  number = {2},
  pages = {455--482},
  year = {2021}
}

@article{barber2023,
  author = {Barber, Rina Foygel and Cand\`es, Emmanuel J. and Ramdas, Aaditya and Tibshirani, Ryan J.},
  title = {Conformal Prediction Beyond Exchangeability},
  journal = {The Annals of Statistics},
  volume = {51},
  number = {2},
  pages = {816--845},
  year = {2023}
}

@inproceedings{gibbs2021,
  author = {Gibbs, Isaac and Cand\`es, Emmanuel J.},
  title = {Adaptive Conformal Inference Under Distribution Shift},
  booktitle = {Advances in Neural Information Processing Systems},
  volume = {34},
  pages = {1660--1672},
  year = {2021}
}

@article{gibbs2024,
  author = {Gibbs, Isaac and Cand\`es, Emmanuel J.},
  title = {Conformal Inference for Online Prediction with Arbitrary Distribution Shifts},
  journal = {Journal of Machine Learning Research},
  volume = {25},
  number = {86},
  pages = {1--36},
  year = {2024}
}

@inproceedings{angelopoulos2023,
  author = {Angelopoulos, Anastasios N. and Cand\`es, Emmanuel J. and Tibshirani, Ryan J.},
  title = {Conformal {PID} Control for Time Series Prediction},
  booktitle = {Advances in Neural Information Processing Systems},
  volume = {36},
  pages = {23047--23074},
  year = {2023}
}

@inproceedings{angelopoulos2024,
  author = {Angelopoulos, Anastasios N. and Bates, Stephen and Fisch, Adam and Lei, Lihua and Schuster, Tal},
  title = {Conformal Risk Control},
  booktitle = {International Conference on Learning Representations},
  year = {2024}
}

@inproceedings{xu2021,
  author = {Xu, Chen and Xie, Yao},
  title = {Conformal Prediction Interval for Dynamic Time-Series},
  booktitle = {Proceedings of the 38th International Conference on Machine Learning},
  series = {PMLR},
  volume = {139},
  pages = {11559--11569},
  year = {2021}
}

@inproceedings{xu2023,
  author = {Xu, Chen and Xie, Yao},
  title = {Sequential Predictive Conformal Inference for Time Series},
  booktitle = {Proceedings of the 40th International Conference on Machine Learning},
  series = {PMLR},
  volume = {202},
  pages = {38707--38727},
  year = {2023}
}

@inproceedings{zaffran2022,
  author = {Zaffran, Margaux and F\'eron, Olivier and Goude, Yannig and Josse, Julie and Dieuleveut, Aymeric},
  title = {Adaptive Conformal Predictions for Time Series},
  booktitle = {Proceedings of the 39th International Conference on Machine Learning},
  series = {PMLR},
  volume = {162},
  pages = {25834--25866},
  year = {2022}
}

@inproceedings{bastani2022,
  author = {Bastani, Osbert and Gupta, Varun and Jung, Christopher and Noarov, Georgy and Ramalingam, Ramya and Roth, Aaron},
  title = {Practical Adversarial Multivalid Conformal Prediction},
  booktitle = {Advances in Neural Information Processing Systems},
  volume = {35},
  pages = {29362--29373},
  year = {2022}
}

@article{guan2023,
  author = {Guan, Leying},
  title = {Localized Conformal Prediction: A Generalized Inference Framework for Conformal Prediction},
  journal = {Biometrika},
  volume = {110},
  number = {1},
  pages = {33--50},
  year = {2023}
}

@inproceedings{ding2023,
  author = {Ding, Tiffany and H\'ebert-Johnson, Ursula and Wang, Ruoxi and Tibshirani, Ryan J.},
  title = {Class-Conditional Conformal Prediction with Many Classes},
  booktitle = {Advances in Neural Information Processing Systems},
  volume = {36},
  year = {2023}
}

@article{hamilton1989,
  author = {Hamilton, James D.},
  title = {A New Approach to the Economic Analysis of Nonstationary Time Series and the Business Cycle},
  journal = {Econometrica},
  volume = {57},
  number = {2},
  pages = {357--384},
  year = {1989}
}

@article{gray1996,
  author = {Gray, Stephen F.},
  title = {Modeling the Conditional Distribution of Interest Rates as a Regime-Switching Process},
  journal = {Journal of Financial Economics},
  volume = {42},
  number = {1},
  pages = {27--62},
  year = {1996}
}

@article{berkowitz2002,
  author = {Berkowitz, Jeremy and O'Brien, James},
  title = {How Accurate Are Value-at-Risk Models at Commercial Banks?},
  journal = {The Journal of Finance},
  volume = {57},
  number = {3},
  pages = {1093--1111},
  year = {2002}
}

@article{basak2001,
  author = {Basak, Suleyman and Shapiro, Alexander},
  title = {Value-at-Risk-Based Risk Management: Optimal Policies and Asset Prices},
  journal = {The Review of Financial Studies},
  volume = {14},
  number = {2},
  pages = {371--405},
  year = {2001}
}

@article{adrian2014,
  author = {Adrian, Tobias and Shin, Hyun Song},
  title = {Procyclical Leverage and Value-at-Risk},
  journal = {The Review of Financial Studies},
  volume = {27},
  number = {2},
  pages = {373--403},
  year = {2014}
}

@article{rockafellar2000,
  author = {Rockafellar, R. Tyrrell and Uryasev, Stanislav},
  title = {Optimization of Conditional Value-at-Risk},
  journal = {Journal of Risk},
  volume = {2},
  number = {3},
  pages = {21--41},
  year = {2000}
}

@article{engle2004,
  author = {Engle, Robert F. and Manganelli, Simone},
  title = {{CAViaR}: Conditional Autoregressive Value at Risk by Regression Quantiles},
  journal = {Journal of Business \& Economic Statistics},
  volume = {22},
  number = {4},
  pages = {367--381},
  year = {2004}
}

@article{kupiec1995,
  author = {Kupiec, Paul H.},
  title = {Techniques for Verifying the Accuracy of Risk Measurement Models},
  journal = {The Journal of Derivatives},
  volume = {3},
  number = {2},
  pages = {73--84},
  year = {1995}
}

@article{christoffersen1998,
  author = {Christoffersen, Peter F.},
  title = {Evaluating Interval Forecasts},
  journal = {International Economic Review},
  volume = {39},
  number = {4},
  pages = {841--862},
  year = {1998}
}

@article{bollerslev1986,
  author = {Bollerslev, Tim},
  title = {Generalized Autoregressive Conditional Heteroskedasticity},
  journal = {Journal of Econometrics},
  volume = {31},
  number = {3},
  pages = {307--327},
  year = {1986}
}

@article{mcneilfrey2000,
  author = {McNeil, Alexander J. and Frey, R{\"u}diger},
  title = {Estimation of Tail-Related Risk Measures for Heteroscedastic Financial Time Series: An Extreme Value Approach},
  journal = {Journal of Empirical Finance},
  volume = {7},
  number = {3--4},
  pages = {271--300},
  year = {2000}
}

@techreport{bcbs2019,
  author = {{Basel Committee on Banking Supervision}},
  title = {Minimum Capital Requirements for Market Risk},
  institution = {Bank for International Settlements},
  year = {2019}
}

@inproceedings{ke2017,
  author = {Ke, Guolin and Meng, Qi and Finley, Thomas and Wang, Taifeng and Chen, Wei and Ma, Weidong and Ye, Qiwei and Liu, Tie-Yan},
  title = {{LightGBM}: A Highly Efficient Gradient Boosting Decision Tree},
  booktitle = {Advances in Neural Information Processing Systems},
  volume = {30},
  year = {2017}
}

@article{moreira2017,
  author = {Moreira, Alan and Muir, Tyler},
  title = {Volatility-Managed Portfolios},
  journal = {The Journal of Finance},
  volume = {72},
  number = {4},
  pages = {1611--1644},
  year = {2017}
}

@inproceedings{malekzadeh2024,
  author = {Malekzadeh, Parvin and Poulos, Zissis and Chen, Jacky and Wang, Zeyu and Plataniotis, Konstantinos N.},
  title = {{EX-DRL}: Hedging Against Heavy Losses with {EXtreme} Distributional Reinforcement Learning},
  booktitle = {Proceedings of the 5th ACM International Conference on AI in Finance},
  year = {2024}
}

@inproceedings{kong2025,
  author = {Kong, Yaxuan and Hwang, Yoontae and Kaiser, M. and Vryonides, Chris and Oomen, Roel and Zohren, Stefan},
  title = {Fusing Narrative Semantics for Financial Volatility Forecasting},
  booktitle = {Proceedings of the 6th ACM International Conference on AI in Finance},
  year = {2025}
}

\end{document}